\documentclass[11pt]{article}

\usepackage{fullpage}
\usepackage{amsmath}
\usepackage{amssymb}
\usepackage{amsthm}
\usepackage{color}
\usepackage{url}
\usepackage{xspace}
\usepackage{paralist}
\usepackage{ifthen}

\usepackage{pgfplots}
\pgfplotsset{width=8cm,compat=1.9}

\newtheorem{theorem}{Theorem}

\newtheorem{lemma}{Lemma}

\newtheorem{proposition}{Proposition}

\def\Pr{{\mathbb{P}}}

\newcommand{\Exp}{\mathbb{E}}
\newcommand{\E}{\Exp}

\newcommand{\inv}[1]{\frac{1}{#1}}

\newcommand{\set}[1]{\left\{ #1 \right\}}

\newcommand{\ignore}[1]{}

\newcommand{\opt}{\textsc{opt}\xspace}

\newcommand{\turan}{\textsc{Tur\'an}\xspace}
\newcommand{\cw}{\textsc{CaroWei}\xspace}
\newcommand{\adeg}{\overline{d}}
\newcommand{\maxalg}{\textsc{Max}\xspace}

\newcommand{\mis}{\textsc{MIS}\xspace}
\newcommand{\mopt}{m_{\opt}}

\newcommand{\MyQuote}[1]{\par\vspace*{0.2cm}\addtocounter{equation}{1}%
     \hspace*{0.5cm}\parbox{12cm}{\em #1}\hfill(\arabic{equation})\vspace*{0.2cm}\\}

\title{\bf Simple and Local Independent Set Approximation}

\author{%
Ravi B. Boppana \\
Department of Mathematics \\
MIT, USA \\
{\tt rboppana@mit.edu}
\and
Magn\'us M. Halld\'orsson%
\thanks{Supported by grants nos.~152679-05 and 174484-05 from the Icelandic Research Fund.} \\
School of Computer Science \\
Reykjavik University, Iceland \\
{\tt mmh@ru.is}
\and
Dror Rawitz%
\thanks{Supported by the Israel Science Foundation (grant no.~497/14).} \\
Faculty of Engineering \\
Bar-Ilan University, Israel \\
{\tt dror.rawitz@biu.ac.il}
}

\begin{document}

\maketitle

\begin{abstract}
We bound the performance guarantees that follow from Tur\'an-like
bounds for unweighted and weighted independent sets in bounded-degree
graphs.  In particular, a randomized approach of Boppana forms a
simple 1-round distributed algorithm, as well as a streaming and
preemptive online algorithm.  We show it gives a tight
$(\Delta+1)/2$-approximation in unweighted graphs of maximum degree
$\Delta$, which is best possible for 1-round distributed algorithms.
For weighted graphs, it gives only a $\Delta$-approximation, but a
simple modification results in an asymptotic expected $0.529
\Delta$-approximation.  This compares with a recent, more complex
$\Delta$-approximation~\cite{BCGS17}, which holds deterministically.
\end{abstract}


\section{Introduction}

Independent sets are among the most fundamental graph structures.
A classic result of Tur\'an~\cite{Turan} says that every graph
$G=(V,E)$ contains an independent set of size at least $\turan(G)
\doteq n/(\adeg+1)$, where $n = |V|$ is the number of vertices and
$\adeg = 2|E|/n$ is the average degree.  Tur\'an's bound is tight for
regular graphs, but for non-regular graphs an improved bound was given
independently by Caro~\cite{c79} and Wei~\cite{w81}:
\begin{equation}
\label{eq:carowei}
\alpha(G) \ge \cw(G) \doteq \sum_{v \in V} \frac{1}{d(v)+1}
~,
\end{equation}
where $\alpha(G)$ is the cardinality of a maximum independent set in
$G$ and $d(v)$ is the degree of vertex $v \in V$.

There are numerous proofs of the Caro-Wei bound, some involving simple
greedy algorithms. Arguably the simplest argument known is a
probabilistic one:
\MyQuote{Uniformly randomly permute the vertices, and output the set
  of vertices that precede all their neighbors in the permutation.}
Each node $v$ precedes its neighbors with probability $1/(d(v)+1)$, so
by linearity of expectation the size of the set output matches exactly
$\cw(G)$. This argument, which first appeared in the book of Alon and
Spencer~\cite{as04}, is due to Boppana~\cite{BoppanaTuran}. It clearly
leads to a very simple local decision rule once the permutation is
selected.

An alternative formulation of the algorithm is practical in certain
contexts.
\MyQuote{Each vertex $v$ picks a random real number $x_v$ from
  $[0,1]$. The vertex joins the independent set if its random number
  is larger than that of its neighbors.}
It suffices to select the numbers with precision $1/n^3$, for which
collisions are very unlikely.

This leads to a fully \emph{1-local} algorithm, in which each node
decides whether to join the independent set after a single round of
communication with its neighbors.  The same $O(\log n)$ bits a node
transmits go to all of its neighbors, which matches the
Broadcast-CONGEST model of distributed algorithms. Furthermore, it is
asynchronous.  This is just about the simplest distributed algorithm
one could hope for.

The simplicity of the approach also allows for other applications. The
basic algorithm works well with edge streams, storing only the
permutation and the current solution as a bit-vector.  The storage can
be reduced with an $\epsilon$-min-wise permutation, at a small cost in
performance.  This can also be viewed as a preemptive online
algorithm, where edges can cause nodes to be kicked out of the
solution but never reenter.

\paragraph*{Our contribution.}
The main purpose of this essay is to analyze the performance guarantees
of Boppana's algorithm on graphs of maximum degree $\Delta$. We show
that it achieves a tight $(\Delta+1)/2$-approximation, which then also
gives a bound on the fidelity of the Caro-Wei bound. In terms of the
average degree $\adeg$, the performance is at most $(\adeg+2)/1.657$.
We also show that the Tur\'an bound has strictly worse performance
than the Caro-Wei bound, but asymptotically the same for
bounded-degree graphs or $(\Delta+1)/2 + 1/(8\Delta)$.

We then address the case of weighted graphs, and find that unchanged
Boppana's algorithm gives only a $(\Delta+1)$-approximation. However,
a slight modification yields an improved approximation which
asymptotically approaches $0.529 (\Delta+1)$.


\subsection{Related work}

Tur\'an~\cite{Turan} showed that $\alpha(G) \ge \turan(G)$.
Caro~\cite{c79} and Wei~\cite{w81} independently showed (in
unpublished technical reports) that $\alpha(G) \ge \cw(G)$.
The bound can also be seen to follow from an earlier work of
Erd\H{o}s~\cite{Erdos:Turan-bound}, who showed that the bound is tight
only for disjoint collections of cliques.  Observe that $\cw(G) \ge
\turan(G)$, for every graph $G$.
%

The min-degree greedy algorithm iteratively adds a minimum-degree node
to the graph, removes it and its neighbors and repeats. It achieves
the Caro-Wei bound~\cite{w81} (see also~\cite{Erdos:Turan-bound}).
Griggs~\cite{Griggs} (see also Chv\'atal and
McDiarmid~\cite{ChvatalMcDiarmid}) showed that the max-degree greedy
algorithm also attains the Caro-Wei bound, where the algorithm
iteratively removes the vertex of maximum degree until the graph is an
independent set.  Sakai et al.~\cite{STY03} analyzed three greedy
algorithms for weighted independent sets and showed them to achieve
certain absolute bounds as well as a $(\Delta+1)$-approximation.

The best sequential approximation known is $\tilde{O}(\Delta/\log^2
\Delta)$,%
\footnote{$\tilde{O}(\cdot)$ suppresses $\log\log n$ factors.}
by Bansal et al.~\cite{BGG15}, which uses semi-definite programming.
This matches the inapproximability result known, up to
doubly-logarithmic factors, that holds assuming the Unique Games
Conjecture~\cite{Austrin09}.  The problem is known to be NP-hard to
approximate within an $O(\Delta/\log^4 \Delta)$ factor~\cite{Chan16}.
For small values of $\Delta$, a
$(\Delta+3)/5$-approximation~\cite{BermanFujito99} is achievable
combinatorially, but requires extensive local search.
%
As for simple greedy algorithms, it was shown in \cite{hr94} that the
performance guarantee of the min-degree greedy algorithm is
$(\Delta+2)/3$, and also pointed out that the max-degree algorithm
attains no better than a $(\Delta+1)/2$ ratio.

Most works on distributed algorithms have focused on finding maximal
independent sets, rather than optimizing their size.  Boppana's
algorithm corresponds to the first of $O(\log n)$ rounds of Luby's
maximal independent set algorithm (see also Alon et al.~\cite{abi86}).
As for approximations, $n^{\Theta(1/k)}$-approximation is achievable
and best possible for local algorithms running in $k$
rounds~\cite{bhkk16}, where the upper bound assumes both unlimited
bandwidth and computation.  Recently, Bar-Yehuda et al.~\cite{BCGS17}
gave a $\Delta$-approximation algorithm for weighted independent sets
using the local ratio technique that runs in time $O(\mis \cdot \log
W)$ rounds, where \mis is the number of rounds needed to compute a
maximal independent set and $W$ is the ratio between the largest and
smallest edge weight.  We improve this approximation ratio by nearly a
factor of~2 using only a single round, but at the price of obtaining a
bound only on expected performance.  Alon~\cite{Alon10} gave nearly
tight bounds for testing independence properties; his lower bound
carries over to distributed algorithms, as we shall see in
Sec.~\ref{sec:lb}.
For matchings, which correspond to independent sets in line graphs,
Kuhn et al.~\cite{kmw16} showed that achieving any constant factor
approximation requires $\Omega(\max(\log \Delta/ \log \log \Delta,
\sqrt{\log n / \log \log n}\,))$ rounds.

Halld\'orsson and Konrad~\cite{hk18} examined how well the Caro-Wei
bound performs in different subclasses of graphs.
They also gave a randomized one-round distributed algorithm where
nodes broadcast only a single bit that yields an independent set of
expected size at least $0.24 \cdot \cw(G)$ on every graph $G$.
This is provably the least requirement for an effective distributed
algorithm, as without degree information, the bounds are polynomially
worse.

Streaming algorithms (including Boppana's) achieving Tur\'an-like
bounds in graphs and hypergraphs were considered in~\cite{HHLS15}, and
streaming algorithms for approximating $\cw(G)$ were given recently by
Cormode et al.~\cite{CDK17}.

Motivated by a packet forwarding application, Emek et
al.~\cite{EHMPRR12} considered the online set packing problem that
corresponds to maintaining strong independent sets of large weight in
hypergraphs under edge additions. We give a tight bound on their
method for the special case of graphs.


\section{Performance of Caro-Wei-Tur\'an Bounds}

We examine here how well the Caro-Wei and the Tur\'an bounds perform
on (unweighted) bounded-degree and sparse graphs.

Let \opt be an optimal independent set of size $\alpha = \alpha(G)$
and let $V' = V \setminus \opt$.
We say that a bound $B(G)$ has a performance ratio $f(\Delta)$ if, for
all graphs $G$ with $\Delta(G) = \Delta$ it holds that $\alpha(G) \ge
B(G) \ge \alpha(G)/f(\Delta)$.


\subsection{Caro-Wei in Bounded-Degree Graphs}

\begin{theorem}
$\cw$ has performance ratio $(\Delta+1)/2$.
\end{theorem}

\begin{proof}
Let $G$ be a graph.
Let $O_i$, for $i = 1, 2, \ldots, \Delta$, denote the number of
vertices in \opt of degree $i$.
Our approach is to separate the contributions of the different $O_i$s
to the Caro-Wei bound.  The nodes of high degree have a smaller direct
contribution, but also have an indirect contribution in forcing more
nodes to be in $V'$.

Let $\mopt$ be the number of edges with an endpoint in \opt.  Each
such edge has the other endpoint in $V'$, whereas
nodes in $V'$ are incident on at most $\Delta$ edges.  Thus,
\begin{equation}
\label{eq:edgecnt}
  \sum_{i=1}^\Delta i \cdot O_i = \mopt \le \Delta |V'| ~. 
\end{equation}
We then obtain
\begin{align*}
\cw(G) = \sum_{v \in V} \frac{1}{d(v)+1} 
& =   \sum_{i=1}^\Delta O_i \cdot \frac{1}{i+1} + \sum_{v \in V'} \frac{1}{d(v)+1} \\
& \ge \sum_{i=1}^\Delta O_i \cdot \frac{1}{i+1} + |V'| \frac{1}{\Delta+1} \\
& \ge \frac{1}{\Delta+1} \sum_{i=1}^\Delta O_i \left( \frac{\Delta+1}{i+1} + \frac{i}{\Delta} \right) 
& & \text{(Applying \eqref{eq:edgecnt})} \\
& = \frac{1}{\Delta+1} \sum_{i=1}^\Delta O_{i} \left( 2 + \frac{\Delta-i}{i+1} - \frac{\Delta-i}{\Delta} \right) \\
& \ge \frac{1}{\Delta+1} \sum_{i=1}^\Delta O_{i} \cdot 2 \\
& = \frac{2}{\Delta+1} \alpha(G)
~,
\end{align*}
obtaining the approximation upper bound claimed.  Observe that the
bound is tight only if all nodes in \opt are of degree $\Delta$ or
$\Delta-1$.

To see that the ratio attained is no better than $(\Delta+1)/2$,
observe that in any regular graph, the algorithm achieves a solution
of exactly $n/(\Delta+1)$, while in bipartite regular graphs the
optimal solution has size $n/2$.
\end{proof}

\paragraph*{Remark.}
Selkow \cite{Selkow94} generalized the Caro-Wei bound by extending
Boppana's algorithm to two rounds. Namely, it adds also the nodes with
no neighbor ordered earlier among those that did not get removed in
the first round. For regular graphs, however, his bound reduces to the
Caro-Wei bound, and thus does not attain a better performance ratio,
given our lower bound construction.


\subsection{Caro-Wei in Sparse Graphs}

We now analyze the performance of the Caro-Wei bound in terms of the
average degree $\adeg = 2|E|/n$.  We shall use a certain application
of the Cauchy-Schwarz inequality, which we state more generality in
hindsight of its application in the following section.

\begin{lemma}
\label{lem:cs}
If $x_1, x_2, \ldots, x_N$ and $w_1, w_2, \ldots, w_N$ are positive reals,
then
\(\displaystyle
\sum_{i=1}^N \frac{w_i^2}{x_i}
\ge \frac{\left(\sum_{i=1}^N w_i\right)^2}{\sum_{i=1}^N x_i}
~.
\)
\end{lemma}

\begin{proof}
The Cauchy-Schwarz inequality implies that
for $u_1, u_2, \ldots, u_N$ and $v_1, v_2, \ldots, v_N$,
\[
\left(\sum_{i=1}^N u_i v_i\right)^2
\le \left(\sum_{i=1}^N u_i^2 \right) \left(\sum_{i=1}^N v_i^2\right)
~.
\]
The claim now follows using $u_i = \sqrt{x_i}$ and $v_i = w_i/\sqrt{x_i}$.
\end{proof}

Note that applying Lemma~\ref{lem:cs} with $w_v = 1$ and $x_v =
d(v)+1$ yields that
\[
\cw(G)
= \sum_{v \in V} \frac{1}{d(v)+1}
\geq \frac{n^2}{\sum_v (d(v)+1)}
=    \frac{n}{\adeg + 1}
=    \turan(G)
~.
\]

\begin{theorem}
$\cw$ has performance ratio at most $(\adeg+2)/1.657$.
\label{lem:adeg-ratio}
\end{theorem}
\begin{proof}
Let $\opt$ be an optimal independent set of size $\alpha = \alpha(G)$
and let $V' = V \setminus \opt$.
Observe that when $|V'| = n - \alpha \ge \alpha$, the Tur\'an bound
gives $n/(\adeg+1) \ge \alpha \cdot 2/(\adeg+1)$, for a performance
ratio of at most $(\adeg+1)/2$.  We assume therefore that $\alpha \ge
\frac{1}{2} n$.

Our approach is to first apply Lemma~\ref{lem:cs} separately on the
parts of $\cw(G)$ corresponding to \opt and $V'$. We then show that
the worst case occurs when all edges cross from \opt to $V'$, indeed
when the graph is bipartite with regular sides.  Optimizing over the
possible sizes of the sides then yields a tight upper and lower
bounds.

Let $\mopt$ denote the number of edges with endpoint in $\opt$,
$m_{V'}$ the number of edges with both endpoints in $V'$ and $m =
\mopt + m_{V'}$ be the total number of edges.  Observe that $\sum_{v
  \in \opt} d(v) = \mopt$ while $\sum_{v \in V'} = \mopt + 2 m_{V'}$.

Lemma~\ref{lem:cs} (with $w_v = 1$ and $x_v = d(v)+1$) applied to \opt
and $V'$ separately yields that
\[
\cw(G)
=   \sum_{v \in \opt} \frac{1}{d(v)+1} + \sum_{v \in V'} \frac{1}{d(v)+1}
\ge \frac{\alpha^2}{\mopt + \alpha} +
    \frac{(n-\alpha)^2}{\mopt + 2 m_{V'} + (n-\alpha)}
~,
\]
%
Denoting $t = \mopt/m$, we get that
\begin{equation}
\cw(G)
\ge \frac{\alpha^2}{t \cdot m + \alpha} + \frac{(n-\alpha)^2}{(2 - t)m + n - \alpha}
~.
\label{eq:post-cs}
\end{equation}
Considered as a function $f$ of $t$, the r.h.s.\ of \eqref{eq:post-cs}
has derivative
\[
\frac{df}{dt}
= - \frac{\alpha^2}{(tm + \alpha)^2} + \frac{(n-\alpha)^2}{((2-t)m + n - \alpha)^2}
~.
\]
Since we assume $\alpha \ge n/2$, it holds that $\alpha^2
(m + n - \alpha)^2 \ge (n-\alpha)^2 (m+\alpha)^2$, and thus $df/dt \le 0$ for all $t
\in [0,1]$.  Hence, denoting $\tau = \alpha/n$, we obtain that
\begin{equation}
\cw(G)
\ge \frac{\alpha^2}{m + \alpha} + \frac{(n-\alpha)^2}{m + n - \alpha}
=   \alpha \left( \frac{\tau}{\adeg/2 + \tau} +
    \frac{(1-\tau)^2/\tau}{\adeg/2+1-\tau} \right)
~.
\label{eq:inv-perf}
\end{equation}
The expression in the parenthesis then upper bounds the reciprocal of
the performance guarantee of $\cw$.

To see that \eqref{eq:inv-perf} is tightest possible, consider
bipartite graphs $G$ with regular sides.  Let $\tau$ be such that
$\tau n$ is the size of the larger side and $q$ is the degree of those
vertices.  Then the number of edges is $m = q \cdot \tau n$, average
degree is $\adeg = 2m/n = 2q\tau$, and the degree of the nodes on the
other side is $m/((1-\tau)n) = \adeg/(2(1-\tau))$.  Clearly
$\alpha(G) = \tau n$, while the Caro-Wei bound gives
\[
\cw(G)
= \frac{\tau n}{\adeg/(2\tau)+1} + \frac{(1 - \tau)n}{\adeg/(2(1-\tau))+1} 
= \alpha(G)
  \left(
    \frac{1}{\adeg/(2\tau)+1} + \frac{(1-\tau)/\tau}{\adeg/(2(1-\tau))+1}
  \right)
~,
\]
which matches \eqref{eq:inv-perf}.

If we round up the lower order terms in the denominator of
\eqref{eq:inv-perf}, we obtain a simpler expression for the asymptotic
performance with $\adeg$:
\[
\cw(G)
\ge \alpha(G) \left( \frac{\tau + (1-\tau)^2/\tau}{\adeg/2 + 1} \right)
~,
\]
which is minimized when $\tau = 1/\sqrt{2}$, for 
a performance ratio at most $(\adeg+2)/(4(\sqrt{2}-1)) \le (\adeg+2)/1.657$.
\end{proof}


\subsection{Tur\'an Bound}

Recall Tur\'an's theorem that $\alpha(G) \ge \turan(G) =
\frac{n}{\adeg + 1} = \frac{n^2}{2|E| + n}$.
%
We find that the guarantee of the Tur\'an bound is strictly weaker than
that of Caro-Wei, yet asymptotically equivalent.

\begin{theorem}
\turan has performance ratio $\displaystyle{\frac{(2 \Delta + 1)^2}{8
    \Delta} = \frac{\Delta + 1}{2} + \frac{1}{8 \Delta} }$.
\end{theorem}

\begin{proof}
Because $\opt = V\setminus V'$ is independent, each of the $|E|$ edges
of~$G$ is incident to at least one vertex in~$V'$.  Conversely, each
vertex in~$V'$ is incident to at most $\Delta$ edges.  So by counting
edges, we get
\[
|E|  \le \Delta |V'| = \Delta(n - \alpha) .
\]
Therefore
\[
2|E| + n \le 2 \Delta (n - \alpha) + n = (2 \Delta + 1) n - 2 \Delta \alpha .
\]
Multiplying by $8 \Delta \alpha$ and using the inequality $4xy \le (x + y)^2$ gives
\[
8 \Delta \alpha (2m + n) 
\le 4 (2 \Delta \alpha) [ (2 \Delta + 1) n - 2 \Delta \alpha ]
\le [(2 \Delta + 1) n]^2
~.
\]
Dividing both sides by $8 \Delta (2m + n)$ gives
\[
\alpha 
\le \frac{(2 \Delta + 1)^2}{8 \Delta} \cdot \frac{n^2}{2m + n}
=   \frac{(2 \Delta + 1)^2}{8 \Delta} \turan(G)
~.
\]

The argument above shows that the performance ratio of Tur\'an's bound
is at most $\frac{(2 \Delta + 1)^2}{8 \Delta}$.  This performance
ratio is tight as a function of~$\Delta$.  To see why, given $\Delta >
0$, let $A$, $B$, and $C$ be disjoint sets of size $2 \Delta - 1$, $2
\Delta - 1$, and $2$, respectively.  Let $G$ be any $\Delta$-regular
bipartite graph with parts $A$ and $B$, together with two isolated
vertices in~$C$.  We can check that $n = 4 \Delta$, $|E| = (2 \Delta -
1) \Delta$, $\turan(G) = \frac{8 \Delta}{2 \Delta + 1}$, and
$\alpha(G) = 2 \Delta + 1$.  So the performance ratio of Tur\'an's
bound on this graph is indeed $\frac{(2 \Delta + 1)^2}{8 \Delta}$.
\end{proof}


\subsection{Limitations of Distributed Algorithms}
\label{sec:lb}

We may assume that we are equipped with unique labels from a universe
of $N$ labels, where $N \ge \Delta \cdot n$. The nodes have knowledge
of $n$, $\Delta$ and $N$, and have unlimited bandwidth and
computational ability.
The nodes have distinct ports for communication with their neighbors,
but do not initially know there labels.

Our result for Boppana's algorithm is optimal for 1-round algorithms.
Observe that the lower bounds below hold also for randomized algorithms.

\begin{theorem}
Every 1-round distributed algorithm has performance ratio at least
$(\Delta+1)/2$, even on unweighted regular graphs.
\end{theorem}

\begin{proof}
In a single round, each node can only learn the labels of their
neighbors and their random bits.

Consider the graph $G_1 = K_{\Delta+1}$, and $G_2$, which is any
$\Delta$-regular bipartite graph.  Distributions over neighborhoods
are identical. Hence, no 1-round algorithm can distinguish between
these graphs.

All nodes will join the independent set with the same probability,
averaged over all possible labelings, since they share the same
views. This probability can be at most $1/(\Delta+1)$, as otherwise
the algorithm would produce incorrect answers on $K_{\Delta+1}$.  The
size of the solution is then at most $n/(\Delta+1)$, while on every
$\Delta$-regular bipartite graphs, the optimal solution contains $n/2$
nodes.
\end{proof}

It is not clear if better results can be obtained when using more rounds.
%
A weaker lower bound holds even for nearly logarithmic number of rounds.

\begin{theorem}
There are positive constants $c_1$ and $c_2$ such that the following
holds: Every $c_1 \log_\Delta n$-round distributed algorithm has
performance ratio at least $c_2 \Delta/\log \Delta$.
\end{theorem}

\begin{proof}
Alon~\cite{Alon10} constructs a $\Delta$-regular graph $G_1$ of girth
$\Omega(\log n/\log \Delta)$ with independence number $O(n/\Delta
\cdot \log \Delta)$, and notes that it is well known that there exist
a bipartite $\Delta$-regular graph $G_2$ of girth $\Omega(\log n/\log
\Delta)$.  The distributions over the $k$-neighborhoods of these
graphs are identical, for $k = O(\log n/\log \Delta)$.  Hence, no
$k$-round distributed algorithm can distinguish between the two.
\end{proof}


\section{Approximations for Weighted Graphs}

In the weighted setting, each node~$v$ is assigned a positive integral
weight $w(v)$ and the objective is to find an independent set $I$
maximizing the total weight $\sum_{v\in I} w(v)$.  For a set $X
\subseteq V$, denote $w(X) = \sum_{x \in X} w(x)$.

Boppana's algorithm can be applied unchanged to weighted graphs,
producing a solution $B$ of expected weight
\[
\E[w(B)] = \sum_{v \in B} w(v) \cdot \frac{1}{d(v)+1}
~,
\]
by linearity of expectation.
This immediately implies that $\E[w(B)] \ge w(V)/(\Delta+1)$, for a
performance ratio at most $\Delta+1$.
To see that this is also the best possible bound, consider the
complete bipartite graphs $K_{N,N}$, where the nodes on one side have
weight 1 and on the other side weight $Q$, for a parameter $Q \ge
\Delta^2$.  The expected weight of the algorithm solution is
$(N+NQ)/(\Delta+1)$, while the optimal solution is of weight $NQ$. The
performance ratio is then $(\Delta+1)/(1 + 1/Q)$, which goes to
$\Delta+1$ as $Q$ gets large.

We therefore turn our attention to modifications that take the weights
into account.


\subsection{Modified algorithm}

We consider now a variation, \maxalg, previously considered in an
online setting in \cite{EHMPRR12}.
\begin{quote}
Each node $v$ picks a random real number $x_v$ uniformly from $[0,1]$.
It broadcasts the values $x_v$ and $w_v$ to its neighbors, who compute
from it $r_v = x_v^{1/w_v}$.  As before, each node $u$ joins the solution 
if its value $r_u$ is the highest among its neighbors.
\end{quote}
The only difference is the computation of $r_v$, which now depends on
the weight $w_v$.  Again the algorithm runs in a single round of
Broadcast-CONGEST, with correctness following as before.
The algorithm was previously shown in \cite{EHMPRR12} to attain a
$\Delta$-approximation.


We obtain a tight bound, which does not have a nice closed expression.

\begin{theorem}
\label{thm:min}
The performance ratio $\rho(\Delta)$ of \maxalg, as a function of $\Delta$, is given by
\[
\frac{1}{\rho}
 = \min_{x \le 1} \left(\frac{x^2}{\Delta+x} + \frac{1}{x\Delta+1}\right)
~.
\]
\end{theorem}

We prove Theorem~\ref{thm:min} in the following subsection.

If we focus on the asymptotics as $\Delta$ gets large, we can ignore
the additive terms in the denominators, obtaining that the performance
ratio approaches
\[
\rho(\Delta)
\underset{\Delta \to \infty}{\to} (\Delta+1) \cdot \inv{x^2 + 1/x}
~.
\]
This is maximized when $x = 2^{-1/3}$ for a ratio of $2^{2/3}(\Delta+1)/3 \sim
(\Delta+1)/1.89 ~\sim 0.529 (\Delta+1)$.
\begin{theorem}
The asymptotic performance ratio of \maxalg is $2^{2/3}(\Delta+1)/3 \sim
0.529 (\Delta+1)$.
\end{theorem}

Figure~\ref{fig:rho} shows $\rho(\Delta)/(\Delta+1)$ as a function of
$\Delta$.
For $\Delta=2$, we find that $1/\rho \sim 0.593$, or $\rho \sim
1.657 \sim 0.562(\Delta+1)$, which is about 6\% larger than $0.529(\Delta+1)$,
but 20\% smaller than $\Delta$.
For $\Delta=1$, the algorithm can made optimal by preferring 
nodes with higher weight than their sole neighbor.

\pgfmathdeclarefunction{f}{2}{%
	\pgfmathparse{(#1 * #1)/(#2 + #1) + 1/(#1 * #2 + 1)}%
}
\pgfmathdeclarefunction{g}{1}{%
	\pgfmathparse{1e10}%
	\global\edef\tmpVal{\pgfmathresult}%
	\foreach \i in {0,0.05,...,1} {%
		\pgfmathparse{min(\tmpVal, f(\i,#1))}%
		\global\edef\tmpVal{\pgfmathresult}%
	}%
	\pgfmathparse{\tmpVal}%
}

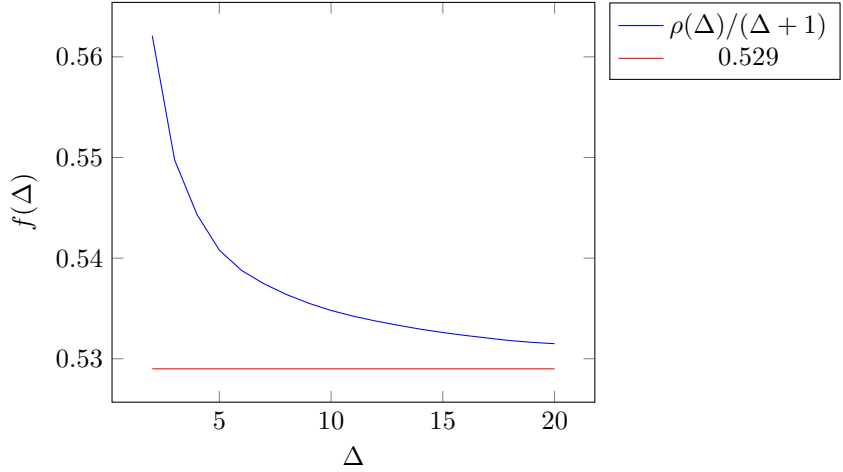
\begin{figure}[t]
\centering
\begin{small}  
\begin{tikzpicture}
\begin{axis}[xlabel={$\Delta$}, ylabel={$f(\Delta)$},
	y tick label style={
		/pgf/number format/fixed,
		/pgf/number format/fixed zerofill,
		/pgf/number format/precision=2 },
        legend pos = outer north east]
	\addplot[blue, samples=19, domain=2:20] { 1/((x+1) * g(x)) };
     	\addlegendentry{$\rho(\Delta)/(\Delta+1)$}
	\addplot[red, samples=19, domain=2:20] { 0.529 ) };
        \addlegendentry{$0.529$}
\end{axis}
\end{tikzpicture}
\end{small}
\caption{Bounds on performance ratio, for small values of $\Delta$.}
\label{fig:rho}
\end{figure}

\subsection{Analysis}

The key property of the \maxalg rule that leads to improved
approximation is that the probability that a node is selected is now
proportional to the fraction of its weight within its closed
neighborhood (consisting of itself and its neighbors). We then obtain
a bound in terms of weights of sets of nodes -- the optimal solution
and the remaining nodes -- using the Cauchy-Schwarz inequality. We
safely upper bound the degree of each node by $\Delta$, but the main
effort then is to show that the worst case occurs when the graph is
bipartite with equal sides. This leads to matching upper and lower
bounds.

Let $N(v)$ denote the set of neighbors of vertex $v$ and $N[v] =
\{v\}\cup N(v)$ its closed neighborhood.
Let \maxalg also refer to the set of nodes selected by \maxalg.

The key property of the \maxalg rule is that the probability that a
node is selected is now proportional to the fraction of its weight
within its closed neighborhood.
We provide a proof for the next lemma for completeness.

\begin{lemma}[\cite{EHMPRR12}]
For each vertex $v \in V$, we have that 
$\displaystyle \Pr[v \in \maxalg] = \frac{w(v)}{w(N[v])} ~.$
\label{lem:ehm}
\end{lemma}

\begin{proof}
Let $r_{\max} = \max\{r_u : u \in N(v)\}$.  By independence of the
random choices we have, for $\alpha \in [0,1]$, that
\[
\Pr[r_{\max} < \alpha]
= \prod_{u \in N(v)} \Pr[r_u < \alpha] 
= \prod_{u \in N(v)} \Pr[x_u < \alpha^{w(u)}] 
= \alpha^{\sum_{u \in N(v)} w(u)} 
= \alpha^{w(N(v))}
~.
\]
It follows that $r_{\max}$ has distribution $D_{w(N(v))}$, where the
distribution $D_z$ has density $f_z(\alpha)=z \alpha^{z-1}$, for
$\alpha \in [0,1]$.  Hence,
\[
\Pr[r_v > r_{\max}]
~= \int_0^1 \Pr[r_{\max} < \alpha] \cdot f_{r_v}(\alpha) d\alpha 
~= \int_0^1 \alpha^{w(N(v))} \cdot w(v) \alpha^{w(v)-1} d\alpha
~=~ \frac{w(v)}{w(N[v])}
~,
\]
as required.
\end{proof}

%


Note that by Lemma~\ref{lem:ehm} and linearity of expectation, we have
that
\begin{equation}
\E[w(S \cap \maxalg)]
= \sum_{v \in S} \Pr[v \in \maxalg] \cdot w(v)
= \sum_{v \in S} \frac{w(v)^2}{w(N[v])}
~,
\label{eq:minsol}
\end{equation}
for any subset $S \subseteq V$.  
Applying Lemma~\ref{lem:cs} (with $x_v = w(N[v])$) gives:
\begin{lemma}
\label{lem:subset}
For any subset $S \subseteq V$ we have that
\(\displaystyle
\E[w(S \cap \maxalg)] \geq \frac{w(S)^2}{\sum_{v \in S} w(N[v])}
~.
\)
\end{lemma}

Applying Lemma~\ref{lem:subset} with $S=V$ gives an absolute lower
bound on the solution size.

\begin{lemma}
\label{lemma:calS}
$\displaystyle \E[w(\maxalg)] 
\geq \frac{w(V)^2}{\sum_{v \in V} w(N[v])}
=    \frac{w(V)^2}{\sum_{v \in V} (d(v)+1) w(v)}
\geq \frac{w(V)}{\Delta+1} ~. $
\end{lemma}

We need the following lemma when showing that worst case occurs for bipartite graphs.

\begin{lemma}
\label{lemma:technical}
Let $a > b > 0$ and let $Z - Y \geq X > 0$.  Then
\[
\min_{t \in [0,1]} \set{\frac{a}{Y + tX} + \frac{b}{Z+(1-t)X}}
= \frac{a}{Y+X} + \frac{b}{Z}
~.
\] 
\end{lemma}
\begin{proof}
Let $f(t) = \frac{a}{Y + tX} + \frac{b}{Z+(1-t)X}$.  We have that
\(
\frac{df(t)}{dt}
= -\frac{aX}{(Y + tX)^2} + \frac{bX}{(Z + (1-t)X)^2}
\),
which is negative for any $t \in [0,1]$, since $a > b$ and $Y + tX \le Z + (1-t)X$.
\end{proof}

Now we are ready to prove Theorem~\ref{thm:min}.

\begin{proof}[Proof of Theorem~\ref{thm:min}]
Let \opt be an optimal solution, and define $V' \doteq V \setminus
\opt$, and $\beta \doteq w(V')/w(\opt)$.
When $\beta \ge 1$, Lemma~\ref{lemma:calS} implies that the
performance ratio is at most $(\Delta+1)/2$.  We therefore focus on
the case where $\beta < 1$.

We first apply Lemma~\ref{lem:subset} separately on \opt and on $V'$,
obtaining:
\begin{equation}
w(\maxalg)
=   w(\maxalg \cap \opt)) + w(\maxalg \cap V')
\ge \frac{w(\opt)^2}{\sum_{v \in \opt} w(N[v])} +
    \frac{w(V')^2}{\sum_{v \in V'} w(N[v])}
~.
\label{eq:alg}
\end{equation}

Let $W = \sum_{v \in V'} w(v) \cdot |N(v) \cap \opt| = \sum_{v \in
  \opt} w(N(v))$ be the weighted degree of the nodes of $V'$ into
$\opt$, which can be viewed as the total of the weights of
neighborhoods of nodes in $\opt$.
Thus, 
\begin{equation}
\sum_{v \in \opt} w(N[v])
= w(\opt) + \sum_{v \in V'} w(v) |N(v) \cap \opt|
= w(\opt) + W
~.
\label{eq:opt1}
\end{equation}
and
\begin{align}
\sum_{v \in V'} w(N[v])
& =    w(V') + \sum_{v \in V'} w(N(v)) \nonumber \\
& =    w(V') + \sum_{v \in \opt} w(v) \cdot |N(v) \cap V'|
             + \sum_{v \in V'} w(v) \cdot |N(v) \cap V'| \nonumber \\
& \leq w(V') + \Delta w(\opt) 
             + \sum_{v \in V'} w(v) \cdot (\Delta - |N(v) \cap \opt|) \nonumber  \\
& =    w(V') + \Delta w(\opt) + \Delta w(V') - W \label{eq:opt2}
~.
\end{align}

Applying \eqref{eq:opt1} and \eqref{eq:opt2} to \eqref{eq:alg} gives
\[
w(\maxalg)
\ge \frac{w(\opt)^2}{w(\opt) + W} +
    \frac{w(V')^2}{w(V') + \Delta w(\opt) + \Delta w(V') - W}
~.
\]
Since $\beta < 1$ and $W \leq \Delta w(V')$ we can use
Lemma~\ref{lemma:technical} with $a = w(\opt)^2$, $b = w(V')^2$, $Y =
w(\opt)$, $Z = w(V') + \Delta w(\opt)$, $X = \Delta w(V')$, and $t =
W/X$.  Hence,
\begin{equation}
w(\maxalg)
\geq \frac{w(\opt)^2}{w(\opt) + \Delta w(V')} +
     \frac{w(V')^2}{w(V') + \Delta w(\opt)}
=    w(\opt) \cdot
     \left( \inv{1 + \Delta \beta} + \frac{\beta^2}{\beta + \Delta} \right)
~.
\label{eq:last}
\end{equation}
The upper bound of the theorem therefore follows.

To see that bound \eqref{eq:last} is tight, consider any
$\Delta$-regular bipartite graph $G=(V,E)$ with $V$ partitioned into
two sets $L$ and $R$, where $|L|=|R|$.  Set the weight of nodes in $L$
and in $R$ as $1$ and $\beta$, respectively, for some $\beta \le 1$.
Clearly, the weight of the optimal solution is $w(\opt) = |L|$.
Observe that
\[
w(\maxalg) 
= |L| \cdot \inv{1 + \Delta \beta} +
  |R| \beta \cdot \frac{\beta}{\beta + \Delta}
= w(\opt) \cdot
  \left( \frac{1}{1+ \beta\Delta} + \frac{\beta^2}{\beta + \Delta} \right)
~,
\]
matching \eqref{eq:last}.
\end{proof}

\paragraph*{Remark.}
Sakai et al.~\cite{STY03} considered the following greedy algorithm
(named GWMIN2): add the vertex $v$ maximizing $w(v)/w(N[v])$ to the
solution, remove its closed neighborhood, and recurse on the remaining
graph. They derived a $(\Delta+1)$-approximation upper bound but not a
matching lower bound.  Since their algorithm attains the bound
\eqref{eq:minsol} (see \cite{STY03}), our analysis implies that it
also attains the bound of Theorem~\ref{thm:min}.


\section{Conclusion}

It's surprising that the best distributed approximations known of
independent sets are obtained by the simplest algorithm.  Repeating
the algorithm on the remaining graph will certainly give a better
solution -- the challenge is to quantify the improvement.

%







\end{document}